\documentclass{ifacconf}

\usepackage{amsmath}
\usepackage{amssymb}
\usepackage{graphicx}
\DeclareGraphicsExtensions{.pdf,.png,.jpg,.eps}

\newtheorem{theorem}{Theorem}
\newtheorem{proof}{Proof}

\newtheorem{remark}{Remark}
\newtheorem{assumption}{Assumption}

\usepackage{graphicx}      
\usepackage{natbib}        
\begin{document}
\begin{frontmatter}

\title{Adaptive Identification of Nonlinear Time-delay Systems Using Output Measurements\thanksref{footnoteinfo}}

\thanks[footnoteinfo]{The results were developed under support of RSF (grant 18-79-10104) in IPME RAS.}

\author[First]{Igor Furtat}
\author[Third]{Yury Orlov}

\address[First]{Institute for Problems of Mechanical Engineering Russian Academy of Sciences, 61 Bolshoy ave V.O., St.-Petersburg, 199178, Russia, (e-mail: cainenash@mail.ru).}
\address[Third]{CICESE Research Center, 3918, Carretera Ensenada-Tijuana, Ensenada, B.C. Mexico, Mexico 22860, (e-mail: yorlov@cicese.mx)}

\begin{abstract}                
A novel adaptive identifier is developed for nonlinear time-delay systems composed of linear, Lipschitz and non-Lipschitz components.
To begin with, an identifier is designed for uncertain  systems with {\it a priori} known delay values, and then it is generalized for systems with unknown delay values.
The algorithm ensures the asymptotic parameter estimation and state observation by using gradient algorithms.
The unknown delays and plant parameters are  estimated by using a special equivalent extension of the plant equation.
The algorithms stability is presented by solvability of linear matrix inequalities.
Simulation results are invoked to support the developed identifier design and to  illustrate the efficiency of the proposed synthesis procedure.
\end{abstract}


\end{frontmatter}

\section{Introduction}

The investigation focuses on adaptive/on-line identification of unknown time-invariant plant parameters.
The existing literature suggests many design methods for plants with lumped model and known structure, see, e.g. \cite{Landau79,Goodwin84,Astrom89,Narendra89,Sastry89,
Ioannou95,Ljung99}.
These methods demonstrate acceptable robustness in the presence of small input and output disturbances or small perturbations of model parameters. Due to this, the methods have found practical applications in electrical vehicle application \cite{Flah14}, robotics \cite{Farza09}, chemical industry  \cite{Ekramian13}, etc.
However, there are only few results applicable to synthesis of plants with time-delays, see, e.g. \cite{Nakagiri95,Verduyn01,Orlov01,Belkoura02,Orlov02,Orlov03,
Orlov09}.

In \cite{Nakagiri95,Verduyn01} the identification of time-delay systems demonstrated complexity of the problem, particularly, the identifiability of a delay system was shown to place a restrictive condition on the structure of the system. This condition was defined through the characteristic matrix of the functional differential equation of the plant whereas no indication was given on how to attain this condition using some accessible inputs.

In \cite{Orlov01,Belkoura02,Orlov02,Orlov03,Orlov09}, the adaptive identifiers were developed step by step, for systems with the complete state information and for single input single output (SISO) linear time delay systems, given in the canonical form of a differential equation of an arbitrary order.  Necessary and sufficient conditions for a linear delay system to be  identifiable have been given in terms of weak
controllability property and nonsmooth input signals. In \cite{Orlov09} the proposed results were experimentally confirmed in an application to a port-fuel-injected internal combustion engine.

However, the identification of single-input single-output (SISO) nonlinear systems with delays has not been addressed so far. Therefore, the main contribution of the paper consists in solving the following problems:

\begin{enumerate}
\item[(i)] design of the adaptive plant identifier for uncertain nonlinear SISO systems with {\it a priori} known time-delays;
\item[(ii)] generalization of the proposed adaptive identifier to the case with unknown time-delays;
\item[(iii)] derivation of the stability conditions in terms of feasibility of linear matrix inequalities (LMIs).
\end{enumerate}

The rest of the paper is outlined as follows. The problem statement is given in Section \ref{sec2}. In Sections \ref{sec3} and \ref{sec4},  two algorithms  are developed side by side for {\it a priori} known and unknown delays, accompanied with the convergence conditions of the proposed algorithms, given in terms of specific LMIs feasibility. In Section \ref{sec5}, the capability  of the proposed synthesis  is illustrated in a simulation study to additionally support the analytical results. Finally, Section \ref{sec6} collects some conclusions.

\textit{Notations}. Throughout the paper, the superscript $\rm T$ stands for the matrix transposition;
$\mathbb R^{n}$ denotes the $n$ dimensional Euclidean space with vector norm $|\cdot|$;
$\mathbb R^{n \times m}$ is the set of all $n \times m$ real matrices;
the notation $P>0$ for $P \in \mathbb R^{n \times n}$ means that $P$ is symmetric and positive definite;
$I$ is the identity matrix of an appropriate dimension;
$diag\{\cdot\}$ is used for a block diagonal matrix.

\section{Problem Formulation} \label{sec2}

Consider a plant model in the form

\begin{equation}
\label{eq2_1}
\begin{array}{l}
\dot{x}(t)=\sum_{i=0}^{k}\big[A_i x(t-\tau_i)+D_i \varphi(x(t-\tau_i))
\\
~~~~~~~~+G_i \psi(y(t-\tau_i))+B_i u(t-\tau_i)\big],
\\
y(t)=Cx(t),
\end{array}
\end{equation}
where $t \geq 0$, $x(t) \in \mathbb R^n$ is the state vector,
$u(t) \in \mathbb R$ is the control input which is assumed to be piece-wise continuous bounded function, 
$y(t) \in \mathbb R$ is the output signal, available for the measurement. For certainty, the time-delay values  $\tau_i$ are ordered as follows $0=\tau_0<\tau_1<...<\tau_r$.

The function $\varphi(x) \in \mathbb R^l$ is globally Lipschitz continuous with an {\it a priori} known Lipschitz constant $L$. The nonlinear
function  $\psi(y(t)) \in \mathbb R^m$ is a piece-wise continuous. 
The well-posedness of system \eqref{eq2_1} is thus ensured in the open-loop.
Along with the above functions, the matrix $C \in \mathbb R^{1 \times n}$ is also known {\it a priori} whereas the matrices $A_i \in \mathbb R^{n \times n}$, $D_i \in \mathbb R^{n \times l}$, $G_i \in \mathbb R^{n \times m}$ and $B_i \in \mathbb R^{n}$ are unknown. Due to the duality of control synthesis and observer design, the measured output is pre-determined with no measurement delays to ensure the identifiability of uncertain matrix parameters (see  Assumption \ref{As03}). Since some matrices might be zero, without loss of generality  system \eqref{eq2_1} has been assumed to possess the same state and input delays.

The delay-free  model \eqref{eq2_1}, formally coming  with  $\tau_1=...=\tau_r=0$, is considered for feedback control and for observation of $x$ in \cite{Farza09,Ekramian13}. In these papers it is noted that such free delay model can describe a number of technical systems and technological processes. For instance,  the estimation of the state and kinetic parameters is addressed in \cite{Farza09} for a bioreactor whereas in \cite{Farza09}, the estimation is investigated for a single-link manipulator with revolute joints actuator. In \cite{Kumar19} the model of chemical and biochemical reactors have input and state delays which arise due to delays in the reception and transmission of data and technological cycles. While controlling electrical equipment, delays are caused by the remote control via digital communication channels. However, for model with delays \eqref{eq2_1} the identification problem has not been addressed so far.

The following technical assumptions are made throughout.
\begin{assumption}
\label{As0} System \eqref{eq2_1} is a BIBO (bounded input - bounded output) system in the sense that
while being driven by a bounded input, the system generates a bounded solution regardless of wherever it is initialized.
\end{assumption}

\begin{assumption}
\label{As1} The input signal $u(t)$ is uniformly bounded and periodic, and  persistently excites system \eqref{eq2_1} in the sense that
there exist constants $C>0$ and $\alpha>0$ such that $\int_{t}^{t+C}\Phi(s)\Phi(s)^{\rm T}ds \geq \alpha I$ with  $\Phi(t)=col\{x(t-\tau_0),...,x(t-\tau_k),\varphi(t-\tau_0),...,\varphi(t-\tau_k),\psi(t-\tau_0),...,\psi(t-\tau_k),u(t-\tau_0),...,u(t-\tau_k)\}$ ,  computed along an arbitrary system trajectory $x(t)$.
\end{assumption}

\begin{assumption}
\label{As2}
The following matching conditions hold
\begin{equation*}
\begin{array}{l}
A_i=A_i^0+T_0 \kappa_i^{A},~~~
D_i=D_i^0+T_0 \kappa_i^{D},~~~
\\
G_i=G_i^0+T_0 \kappa_{i}^{G},~~~
B_i=B_i^0+T_0 \kappa_{i}^{B},~~~i=0,...,k,
\end{array}
\end{equation*}
where $A_i^0$, $D_i^0$, $G_i^0$, $B_i^0$, and $T_{0} \in \mathbb R^{n}$ are known and $CT_0 \neq 0$, whereas $\kappa_i^{A} \in \mathbb R^{1 \times n}$, $\kappa_i^{D} \in \mathbb R^{1 \times l}$, $\kappa_{i}^{G} \in \mathbb R^{1 \times m}$, and $\kappa_{i}^{B} \in \mathbb R$ are unknown.
\end{assumption}

\begin{assumption}
\label{As03}
System \eqref{eq2_1} is identifiable in the sense that there exists a persistently exciting input $u(t)$ such that the unknown parameters in \eqref{eq2_1} are uniquely determined from the measured output $y(t)$ \cite{Orlov03}.
\end{assumption}

\begin{assumption}
\label{As04}
System \eqref{eq2_1} is locally observable in the sense that the difference $\Delta x(t)$ of arbitrary solutions $x(t),\hat{x}(t)$ of \eqref{eq2_1} asymptotically escapes $\lim_{t \rightarrow \infty} \Delta x(t)=0$ to zero provided that these solutions generate the same output  $Cx(t)=C(\hat{x}(t)$ for all $t\geq 0$.
\end{assumption}

The above assumptions are made for technical reasons. Assumption \ref{As0} is well-recognized from the linear theory  to be imposed on a system for its on-line identification in open-loop \cite{Orlov02}.

Assumption \ref{As1} is an extension of the well-known Persistency-of-Excitation (PE) condition (see definition of PE condition in \cite{Shimkin87,Mareels88,Ioannou96}) to the underlying time-delay system. Such an assumption is typically invoked to prove the identifier convergence to the nominal system parameters (cf. that of Theorem \ref{Th1} where the input periodicity is particularly utilized to apply the invariance principle).

Assumption \ref{As2} is inspired from a finite-dimensional matching condition counterpart used to ensure the identifiability of the unknown parameters. A similar identifiability problem is repeatedly discussed in the adaptive control \cite{Tao03,Hovakimyan10} and adaptive identification of free-delay linear plants in \cite{Tao03}.

Assumptions \ref{As03} and \ref{As04}, coupled together, ensure that relation
\begin{equation}
\label{eq3_11}
\begin{array}{l}
\lim_{t \rightarrow \infty}C T_0 \sum_{i=0}^{k}\Big[
 \Delta \kappa_i^{A} x(t-\tau_i)
+ \Delta \kappa_i^{D} \varphi(x(t-\tau_i))
\\
+ \Delta \kappa_i^{G} \psi(y(t-\tau_i))

+ \Delta \kappa_i^{B} u(t-\tau_i)
\Big]=0
\end{array}
\end{equation}
can only be satisfied for the trivial parameter errors
\begin{equation}
\begin{array}{l}
 \Delta\kappa_i^{A}=0,\ \Delta\kappa_i^{D}=0,\ \Delta\kappa_i^{G}=0,\
\\
\Delta\kappa_i^{B}=0,\ i=0,1,\ldots,k
\label{AlgebraicDeltai}
\end{array}
\end{equation}
where  $\Delta \kappa_i^{A}=\kappa_i^{A}-\hat{\kappa}_i^{A}$,
$\Delta \kappa_i^{D}=\kappa_i^{D}-\hat{\kappa}_i^{D}$,
$\Delta \kappa_i^{G}=\kappa_i^{G}-\hat{\kappa}_i^{G}$,
$\Delta \kappa_i^{B}=\kappa_i^{B}-\hat{\kappa}_i^{B}$, $i=0,...,k,$ are the deviations of the nominal parameters $\kappa_i^{A}$, $\kappa_i^{D}$, $\kappa_i^{G}$,  $\kappa_i^{B}$ from their estimates
$\hat\kappa_i^{A}$, $\hat\kappa_i^{D}$, $\hat\kappa_i^{G}$, $\hat\kappa_i^{B}$.
To reproduce this conclusion it suffices to equate the outputs $Cx(t)=C\hat{x}(t)$ of system \eqref{eq2_1}, generated with
the nominal parameters $\kappa_i^{A}$, $\kappa_i^{D}$, $\kappa_i^{G}$,  $\kappa_i^{B}$ and their estimates
$\hat\kappa_i^{A}$, $\hat\kappa_i^{D}$, $\hat\kappa_i^{G}$, $\hat\kappa_i^{B}$, and after that
differentiate the resulting equality along the corresponding solutions of  \eqref{eq2_1}, taking into account the local observability of the system.

If confined to SISO time-delay systems, Assumption \ref{As03} is well-known  \cite{Orlov09} to hold true.  The identifiability of the system parameters and delays can then be enforced by applying to the system a sufficiently nonsmooth signal that persistently excites the system.  These signals are constructively introduced by imposing the state of the system and the system input to have different
smoothness properties \cite{Orlov03}.  In general,
 Assumption \ref{As03}, roughly speaking,   requires that not only  the solutions $x(t-\tau_i)$ and the inputs $u(t-\tau_i)$, but in addition to  \cite{Orlov03}, also  $\varphi(t-\tau_i)$ and $\psi(t-\tau_i)$, viewed in combination with $x(t-\tau_i)$ and $u(t-\tau_i)$,  present different behaviour. For MIMO systems, this topic however calls for further investigation and remains beyond the scope of the paper.

In the sequel, Assumption \ref{As03} is simply postulated, and only numerical evidences are given in  Section \ref{sec5} to support it in a nontrivial academic example, illustrating the theory developed.

For later use, let us introduce the estimation errors
\begin{equation*}
\begin{array}{l}
\Delta \kappa_i^{A}(t)=\kappa_i^{A}-\hat{\kappa}_i^{A}(t),~~~
\Delta \kappa_i^{D}(t)=\kappa_i^{D}-\hat{\kappa}_i^{D}(t),~~~
\\
\Delta \kappa_i^{G}(t)=\kappa_i^{G}-\hat{\kappa}_i^{G}(t),~
\Delta \kappa_i^{B}(t)=\kappa_i^{B}-\hat{\kappa}_i^{B}(t),~ i=0,...,k,
\\
\varepsilon(t)=x(t)-\hat{x}(t),
\end{array}
\end{equation*}
where $\hat\kappa_i^{A}(t)$, $\hat\kappa_i^{D}(t)$, $\hat\kappa_i^{G}(t)$, $\hat\kappa_i^{B}(t)$, and $\hat x(t)$ are dynamic estimates of the nominal  values $\kappa_i^{A}$, $\kappa_i^{D}$, $\kappa_i^{G}$, $\kappa_i^{B}$, and $x(t)$ accordingly.

The objective is to design an identification algorithm that ensures

\begin{equation}
\label{eq2_3}
\begin{array}{l}
\lim\limits_{t \to \infty}\Delta\kappa_i^{A}(t)=0,~~
\lim\limits_{t \to \infty}\Delta\kappa_i^{D}(t)=0,~~
\\
\lim\limits_{t \to \infty}\Delta\kappa_i^{G}(t)=0,~~
\lim\limits_{t \to \infty}\Delta\kappa_i^{B}(t)=0,~~i=0,...,k,
\\
\lim\limits_{t \to \infty}\varepsilon(t)=0.
\end{array}
\end{equation}
In what follows, such an identification algorithm is developed for the nonlinear time-delay system in question.

\section{Adaptive identifier design under {\it a priori} known delay values} \label{sec3}

Consider a plant model
\begin{equation}
\label{eq3_1}
\begin{array}{l}
\dot{\hat{x}}(t)=\sum_{i=0}^{k}\Big[
A_i^0 \hat{x}(t-\tau_i)+D_i^0 \varphi(\hat{x}(t-\tau_i))
\\
~~~~~~~~+G_i^0 \psi(y(t-\tau_i))+B_i^0 u(t-\tau_i)-Y_i \varepsilon(t-\tau_i)
\\
~~~~~~~~+T_0\sum_{i=0}^{k}\Big[\hat{\kappa}_i^{A}(t) \hat{x}(t-\tau_i)+ \hat{\kappa}_i^{D}(t) \varphi(\hat{x}(t-\tau_i))
\\
~~~~~~~~+ \hat{\kappa}_i^{G}(t) \psi(y(t-\tau_i))
+ \hat{\kappa}_i^{B}(t) u(t-\tau_i)\big)\Big],
\\
\hat{y}(t)=C\hat{x}(t),
\end{array}
\end{equation}
of the same structure as that of \eqref{eq2_1} with Hurwitz matrices $Y_i \in \mathbb R^{n \times n}$ at the  designer disposition. Let the model parameters be updated as
$\dot{\hat \kappa}_i^{A}(t)^{\rm T}=-\Gamma_i^{A} e(t)\hat{x}(t-\tau_i)$,
$\dot{\hat \kappa}_i^{D}(t)^{\rm T}=-\Gamma_i^{D} e(t)\varphi(\hat{x}(t-\tau_i))$, $\dot{\hat \kappa}_i^{G}(t)-\Gamma_i^{G} e(t)\psi(y(t-\tau_i))$,
$\dot{\hat \kappa}_i^{B}(t)=-\Gamma_i^{B} e(t)u(t-\tau_i)$, $i=0,1,\ldots,k$,
so that the parameter errors are governed by
\begin{equation}
\label{eq_Th02}
\begin{array}{l}
\Delta\dot{\kappa}_i^{A}(t)^{\rm T}=-\Gamma_i^{A} e(t)\hat{x}(t-\tau_i),~~~~~
\\
\Delta\dot{\kappa}_i^{D}(t)^{\rm T}=-\Gamma_i^{D} e(t)\varphi(\hat{x}(t-\tau_i)),
\\
\Delta\dot{\kappa}_i^{G}(t)^{\rm T}=-\Gamma_i^{G} e(t)\psi(y(t-\tau_i)),~~~~~
\\
\Delta\dot{\kappa}_i^{B}(t)=-\Gamma_i^{B} e(t)u(t-\tau_i).
\end{array}
\end{equation}
The matrices $\Gamma_i^{A}$, $\Gamma_i^{D}$,  $\Gamma_i^{G}$, and $\Gamma_i^{B}>0$ are positive definite and of appropriate dimensions. Then the plant deviation  $\varepsilon(t)$ from the model variable is computed according to \eqref{eq2_1} and \eqref{eq3_1}, and it is therefore  governed by
\begin{equation}
\label{eq3_2}
\begin{array}{l}
\dot{\varepsilon}(t)=\sum_{i=0}^{k}\Big[A_i \varepsilon(t-\tau_i)
+D_i [\varphi(x(t-\tau_i)
\\
~~~-\varphi(\hat{x}(t-\tau_i)]-Y_i \varepsilon(t-\tau_i)
\\
~~~+T_0 \sum_{i=0}^{k}\Big[
\Delta \kappa_i^{A}(t) \hat{x}(t-\tau_i)
+\Delta \kappa_i^{D}(t) \varphi(\hat{x}(t-\tau_i))
\\
~~~+ \Delta \kappa_i^{G}(t) \psi(y(t-\tau_i))
+ \Delta \kappa_i^{B}(t) u(t-\tau_i)
\Big],
\\
e(t)=C\varepsilon(t).
\end{array}
\end{equation}

The result, stated below, relies on the notation
\begin{equation}
\label{Matrix_Notation}
\begin{array}{l}
\bar{\Psi}_{11}=A_0^{\rm T}P+PA_0-Y_0+\sum_{i=0}^{k}S_i,

\\

\Psi_{11}=
\begin{bmatrix}
\bar{\Psi}_{11} & P(A_1-Y_1) &...& P(A_k-Y_k)
\\
* & -S_1-Y_1 &...& 0
\\
\vdots & \vdots & \ddots & \vdots
\\
* & * & ... & -S_k-Y_k
\end{bmatrix},

\\

\Psi_{12}=
\begin{bmatrix}
PD_0 & PD_1 &...& PD_k
\\
* & 0 &...& 0
\\
\vdots & \vdots & \ddots & \vdots
\\
* & * & ... & 0
\end{bmatrix},

\\

\Psi=
\begin{bmatrix}
\bar{\Psi}_{11}+L^2I & \Psi_{12}
\\
* & -I
\end{bmatrix}.
\end{array}
\end{equation}
Here the notation $"*"$ means a symmetric block of a symmetric matrix.

\begin{theorem}
\label{Th1}
Let the delay values $\tau_j$, $j=1,...,k$ be known \textit{a priori}, and let  Assumptions \ref{As0}--\ref{As04} hold. Moreover, let  there exist matrices $P=P^{\rm T}>0$, $S_i>0$, $i=0,...,k$ such that the relations
\begin{equation}
\label{eq_Th01}
\begin{array}{l}
\Psi < 0
~~~\mbox{and}~~~
PT_0=C^{\rm T}
\end{array}
\end{equation}
hold true.
Then the over-all error system \eqref{eq_Th02}, \eqref{eq3_2} is asymptotically stable so that the above objective \eqref{eq2_3} is achieved with identifier \eqref{eq3_1}, updated according to \eqref{eq3_2}.

\end{theorem}

\begin{proof}The proof is constructed in two steps.

\subsection{Stability analysis} \label{SA1}

Consider Lyapunov-Krasovskii functional

\begin{equation}
\label{eq3_3}
\begin{array}{c}
V=V_1+V_2,
\end{array}
\end{equation}
where

\begin{equation}
\label{eq3_4}
\begin{array}{c}
V_1=\varepsilon^{\rm T}(t)P\varepsilon(t)
+\sum_{i=0}^{k}\Big[\Delta \kappa_i^{A}(t)(\Gamma_i^{A})^{-1} \kappa_i^{A}(t)^{\rm T}
\\
~~~~~+\Delta \kappa_i^{D}(t)(\Gamma_i^{D})^{-1} \Delta \kappa_i^{D}(t)^{\rm T}
\\
~~~~~+ \Delta \kappa_i^{G}(t)(\Gamma_i^{G})^{-1}\Delta \kappa_i^{G}(t)^{\rm T}
\\
~~~~~+(\Gamma_i^{B})^{-1}[\Delta \kappa_i^{B}(t)]^2
\Big],
\end{array}
\end{equation}


\begin{equation}
\label{eq3_5}
\begin{array}{c}
V_2=\sum_{i=0}^{k}\int_{t-\tau_i}^{t}\varepsilon^{\rm T}(s)S_i\varepsilon(s)ds.
\end{array}
\end{equation}

The computation of the time-derivative of $V_1$ along the trajectories of 
\eqref{eq_Th02} and \eqref{eq3_2} yields

%
%
%

\begin{equation}
\label{eq3_7}
\begin{array}{l}
\dot{V}_1
=\varepsilon^{\rm T}(A_0^{\rm T}P+PA_0-Y_0)\varepsilon^{\rm T}
\\
~~~~~~+2\varepsilon^{\rm T}PD_0 [\varphi(x(t)-\varphi(\hat{x}(t))]
\\
~~~~~~+2\varepsilon^{\rm T}P\sum_{i=1}^{k}\Big[A_i \varepsilon(t-\tau_i)
\\
~~~~~~+D_i [\varphi(x(t-\tau_i)-\varphi(\hat{x}(t-\tau_i)]
\\
~~~~~~-Y_i \varepsilon(t-\tau_i)\Big].
\end{array}
\end{equation}

In turn, computing the time-derivative of $V_2$, yields
\begin{equation}
\label{eq3_8}
\begin{array}{c}
\dot{V}_2
=\sum_{i=0}^{k}[\varepsilon(t)^{\rm T}S_i\varepsilon(t)-\varepsilon(t-\tau_i)^{\rm T}S_i\varepsilon(t-\tau_i)].
\end{array}
\end{equation}
Introducing the vectors $\chi_1(t)=col\{\varepsilon(t),\varepsilon(t-\tau_1),...,\varepsilon(t-\tau_k)\}$, $\chi_2(t)=col\{\varphi(x(t))-\varphi(\hat{x}(t)),\varphi(x(t-\tau_1))-\varphi(\hat{x}(t-\tau_1)),...,\varphi(x(t-\tau_k))-\varphi(\hat{x}(t-\tau_k))\}$ and combining \eqref{Matrix_Notation}, \eqref{eq3_7} and \eqref{eq3_8}, let us represent $\dot{V}$ in the form

\begin{equation}
\label{eq3_9}
\begin{array}{c}
\dot{V}
=[\chi_1^{\rm T}~ \chi_2^{\rm T}]
\begin{bmatrix}
\Psi_{11} & \Psi_{12}
\\
* & 0
\end{bmatrix}
\begin{bmatrix}
\chi_1
\\
\chi_2
\end{bmatrix}
\end{array}
\end{equation}
where $\Psi_{11}, \Psi_{12}$ are governed by \eqref{Matrix_Notation}.

Since the right-hand side of \eqref{eq3_9}  does not depend of the estimation  errors  it  cannot be negative definite, however it might be negative semi-definite.
In order to conclude that $\dot{V} \leq 0$ it suffices to establish that the matrix $\begin{bmatrix}
\Psi_{11} & \Psi_{12}
\\
* & 0
\end{bmatrix}$ is negative definite. For reproducing this, let us deduce  the inequality
\begin{equation}
\label{eq3_10}
\begin{array}{c}
[\chi_1^{\rm T}~ \chi_2^{\rm T}]
\mbox{diag} \{ L^2I,-I\}
[\chi_1^{\rm T}~ \chi_2^{\rm T}]^{\rm T}
\geq 0.
\end{array}
\end{equation}
 from the global Lipschitz condition
$[\varphi(x(t-\tau_i))-\varphi(\hat{x}(t-\tau_i))]^{\rm T} [\varphi(x(t-\tau_i))-\varphi(\hat{x}(t-\tau_i))] \leq L^2 \varepsilon(t-\tau_i)^{\rm T}\varepsilon(t-\tau_i)$, $i=0,...,k$,
imposed in Section \ref{sec2} on the  function  $\varphi(x)$.
Using S-procedure from \cite{Yakubovich73,Boyd04} and taking  \eqref{eq3_10} into account, let us represent \eqref{eq3_9} as $\dot{V} \leq [\chi_1^{\rm T}~ \chi_2^{\rm T}] \Psi [\chi_1^{\rm T}~ \chi_2^{\rm T}]^{\rm T}$ where $\Psi$ is given by \eqref{Matrix_Notation}.
It is straightforward now to conclude that the inequality $\dot{V} \leq 0$ holds provided that  $\Psi < 0$ which is actually guaranteed by \eqref{eq_Th01}.

It follows that the signals $\varepsilon(t)$, $\Delta\kappa_i^{A}(t)$, $\Delta\kappa_i^{D}(t)$, $\Delta\kappa_i^{G}(t)$, and $\Delta\kappa_i^{B}(t)$ are uniformly bounded, and the over-all error system \eqref{eq_Th02}, \eqref{eq3_2} is stable.
Moreover, taking into account Assumption \ref{As0}, the boundedness of $\hat{x}(t)$ follows from that of  $x(t)$ and  $\varepsilon(t)$ as well as  the boundedness of $\dot\varepsilon(t)$  is straightforwardly concluded from \eqref{eq3_2} due to the boundedness of $\varepsilon(t)$ and $x(t)$.

\subsection{Asymptotic stability analysis} \label{BP1}

The asymptotic stability of the error system \eqref{eq_Th02}, \eqref{eq3_2} is established based on the infinite-dimensional extension \cite{Hen81} of the Krasovskii--LaSalle invariance principle  to time-periodic delay systems, similar to that of \cite{Rouche}. According to the invariance principle, thus extended, there must be a convergence of the trajectories of the error system \eqref{eq_Th02}, \eqref{eq3_2} to the largest invariant subset of the set of the solutions of \eqref{eq_Th02}, \eqref{eq3_2} for which  $\dot V =0$, or equivalently
\begin{equation}
\chi_1\equiv 0,\ \chi_2\equiv 0.
\label{lis}
\end{equation}

Let us show that manifold \eqref{lis} does not contain nontrivial trajectories of \eqref{eq_Th02}, \eqref{eq3_2}. Indeed, if confined to \eqref{lis}, one has
\begin{equation}
\varepsilon\equiv 0\ \Rightarrow  \ \dot\varepsilon\equiv 0,
\label{RA}
\end{equation}
and by virtue of  \eqref{eq_Th02}, one derives that
\begin{equation}
\begin{array}{l}
\Delta\dot\kappa_i^{A}(t)=0,\ \Delta\dot\kappa_i^{D}(t)=0,\ \Delta\dot\kappa_i^{G}(t)=0,\
\\
\Delta\dot\kappa_i^{B}(t)=0,\ i=0,1,\ldots,k.
\label{Deltai}
\end{array}
\end{equation}
Then along the invariant subset \eqref{lis}, relation $\sum_{i=0}^{k}\Big[A_i \varepsilon(t-\tau_i)
+D_i [\varphi(x(t-\tau_i)-\varphi(\hat{x}(t-\tau_i)]-Y_i \varepsilon(t-\tau_i)
-T_0 \Delta \kappa_i^{A}(t) \varepsilon(t-\tau_i)
-T_0 \Delta \kappa_i^{D}(t) [\varphi(x(t-\tau_i))-\varphi(\hat{x}(t-\tau_i))]
\Big] = 0$ is straightforwardly verified. With this in mind and taking relations \eqref{RA}, \eqref{Deltai} into account, the error dynamics \eqref{eq3_2} result in \eqref{eq3_11}, thereby ensuring that \eqref{AlgebraicDeltai} holds true. Thus, the largest invariant subset of the set   $\dot V =0$ coincides with the origin, and by applying the invariance principle, the  error system \eqref{eq_Th02}, \eqref{eq3_2} is established to be asymptotically stable.
This completes the proof of Theorem \ref{Th1}.

%
%
%

\end{proof}

\section{Case of unknown time-delays} \label{sec4}

In the present section, the number $k$ of time-delays $\tau_i$, $i=1,...,k$ of the plant dynamics \eqref{eq2_1} are no longer  assumed to be known {\it a priori}. The  identifier design in such a frame calls for another interpretation of equation \eqref{eq2_1}. To formally apply the developed identifier let us introduce the following notations

\begin{equation}
\label{eq4_00}
\begin{array}{l}
\bar{k} \geq k,~~~
0=\hat{\tau}_0<\hat{\tau}_1<...<\hat{\tau}_{\bar{k}},~~~
\\
\bar{A}_i \in \mathbb R^{n \times n}, \bar{D}_i \in \mathbb R^{n \times l},
\\
\bar{G}_i  \in \mathbb R^{n \times m}, \bar{B}_i  \in \mathbb R^{n},~~ i=1,...,\bar{k},
\\
\Xi=\{\tau_1,...,\tau_k\},
\\
\bar{\Xi}=\{\hat{\tau}_1,...,\hat{\tau}_{\bar{k}}\},
\\
\Lambda=\{A_i,D_i,G_i,B_i,~ i=1,...,k\},
\\
\bar{\Lambda}=\{\bar{A}_i, \bar{D}_i, \bar{G}_i, \bar{B}_i,~ i=1,...,\bar{k}\},
\end{array}
\end{equation}
and impose the following assumptions.

\begin{assumption}
\label{As3}
The values of $\bar{k}$ and $\hat{\tau}_i$, $i=1,...,\bar{k}$ are known {\it a priori} whereas the matrices $\bar{A}_i$, $\bar{D}_i$, $\bar{G}_i$, $\bar{B}_i$, $i=1,...,\bar{k}$ are unknown.
\end{assumption}

\begin{assumption}
\label{As4}
The implications $\Xi \subseteq \bar{\Xi}$ and $\Lambda \subseteq \bar{\Lambda}$ are in force and the sets $\bar{\Xi} \setminus \Xi$ and $\bar{\Lambda} \setminus \Lambda$ contain zero elements.
\end{assumption}

The above assumptions presume that unknown plant delays belong to an {\it a priori} known finite set as it happens, e.g., in computer networks where transmission delays are commensurate a specific precision. Thus, the identification of unknown delay values is reduced to identifying fictitious delay values, which are associated with zero matrix multipliers to be identified along with other nonzero parameter values. Indeed, using notations \eqref{eq4_00} and Assumptions \ref{As3}, \ref{As4}, rewrite plant equation \eqref{eq2_1} in the form

\begin{equation}
\label{eq4_0}
\begin{array}{l}
\dot{x}(t)=\sum_{i=0}^{\bar{k}}\big[\bar{A}_i x(t-\hat{\tau}_i)+\bar{D}_i \varphi(x(t-\hat{\tau}_i))
\\
~~~~~~~~+\bar{G}_i \psi(y(t-\hat{\tau}_i))+\bar{B}_i u(t-\hat{\tau}_i)\big],
\\
y(t)=Cx(t).
\end{array}
\end{equation}

It is worth noticing that  model \eqref{eq4_0} has been obtained based on the modifications of  Assumptions \ref{As1} and \ref{As2},  given below.

\begin{assumption}
\label{As5}
The input signal $u(t)$ is uniformly bounded and periodic, and  persistently excites system \eqref{eq4_0} in the sense that
there exist constants $C>0$ and $\alpha>0$ such that $\int_{t}^{t+C}\Phi(s)\Phi(s)^{\rm T}ds \geq \alpha I$ with  $\Phi(t)=col\{x(t-\bar{\tau}_0),...,x(t-\bar{\tau}_{\bar{k}}),\varphi(t-\bar{\tau}_0),...,\varphi(t-\bar{\tau}_{\bar{k}}),\psi(t-\bar{\tau}_0),...,\psi(t-\bar{\tau}_{\bar{k}}),u(t-\bar{\tau}_0),...,u(t-\bar{\tau}_{\bar{k}})\}$,  computed along an arbitrary system trajectory $x(t)$.
\end{assumption}

\begin{assumption}
\label{As6}
The following matching conditions hold
$\bar{A}_i=\bar{A}_i^0+T_{0} \kappa_i^{\bar{A}}$,
$\bar{D}_i=\bar{D}_i^0+T_{0} \kappa_i^{\bar{D}}$,
$\bar{G}_i=\bar{G}_i^0+T_{0} \kappa_{i}^{\bar{G}}$,
$\bar{B}_i=\bar{B}_i^0+T_{0} \kappa_{i}^{\bar{B}}$, $i=0,...,\bar{k}$,
where $\bar{A}_i^0$, $\bar{D}_i^0$, $\bar{G}_i^0$, $\bar{B}_i^0$, $T_{0} \in \mathbb R^{n}$
 are known matrices and vectors, and $CT_0 \neq 0$, whereas
$\kappa_i^{\bar{A}} \in \mathbb R^{n \times 1}$,
$\kappa_i^{\bar{D}} \in \mathbb R^{1 \times l}$,
$\kappa_{i}^{\bar{G}} \in \mathbb R^{1 \times m}$, and
$\kappa_{i}^{\bar{B}} \in \mathbb R$ are unknown.
\end{assumption}

The basic idea behind  the representation of model \eqref{eq2_1} in form \eqref{eq4_0} is as follows. If $x(t-\hat{\tau}_l)=x(t-\tau_j)$ for some $l \in \{i,...,\bar{k}\}$ and $j \in \{i,...,k\}$, then $\bar{A}_l=A_j$. Otherwise, $x(t-\hat{\tau}_l) \neq x(t-\tau_j)$ for any $l \in \{i,...,\bar{k}\}$ and $j \in \{i,...,k\}$, and $\bar{A}_l=0$. Similar comments are also in order for  other terms in \eqref{eq4_0}. Thus, identifying nonzero matrices among of $\bar{A}_i, \bar{D}_i, \bar{G}_i$, $\bar{B}_i, i=1,...,\bar{k}$ yields  corresponding (non-fictitious) time-delays.

Let us now consider the identifier in the form

\begin{equation}
\label{eq4_2}
\begin{array}{l}
\dot{\hat{x}}(t)=\sum_{i=0}^{\bar{k}}\Big[
\bar{A}_i^0 \hat{x}(t-\tau_i)+\bar{D}_i^0 \varphi(\hat{x}(t-\tau_i))
\\
~~~~~~~~~+\bar{G}_i^0 \psi(y(t-\tau_i))+\bar{B}_i^0 u(t-\tau_i)\Big]
\\
~~~~~~~~~+T_0 \sum_{i=0}^{\bar{k}}\Big[ \hat{\kappa}_i^{A}(t) \hat{x}(t-\tau_i)
\\
~~~~~~~~~+\hat{\kappa}_i^{D}(t) \varphi(\hat{x}(t-\tau_i))
+ \hat{\kappa}_i^{\bar{A}}(t) \hat{x}(t-\hat{\tau}_i)
\\
~~~~~~~~~+\hat{\kappa}_i^{\bar{D}}(t) \varphi(\hat{x}(t-\hat{\tau}_i))
+\hat{\kappa}_i^{\bar{G}}(t) \psi(y(t-\hat{\tau}_i))
\\
~~~~~~~~~+ \hat{\kappa}_i^{\bar{B}}(t) u(t-\hat{\tau}_i)\Big] -Y_i \varepsilon(t-\hat{\tau}_i),
\\
\hat{y}(t)=C\hat{x}(t),
\end{array}
\end{equation}

Computing the time derivative of $\varepsilon(t)=x(t)-\hat{x}(t)$ along the trajectories \eqref{eq4_0} and \eqref{eq4_2}, one obtains

\begin{equation}
\label{eq4_3}
\begin{array}{l}
\dot{\varepsilon}(t)=\sum_{i=0}^{\bar{k}}\Big[\bar{A}_i \varepsilon(t-\hat{\tau}_i)-Y_i \varepsilon(t-\hat{\tau}_i)
\\
~~~~~~~~+\bar{D}_i [\varphi(x(t-\hat{\tau}_i)-\varphi(\hat{x}(t-\hat{\tau}_i)]
\Big]
\\
~~~~~~~~+T_0 \sum_{i=0}^{\bar{k}}\Big[
\Delta \kappa_i^{\bar{A}}(t) \hat{x}(t-\hat{\tau}_i)
\\
~~~~~~~~+\Delta \kappa_i^{\bar{D}}(t) \varphi(\hat{x}(t-\hat{\tau}_i))
\\
~~~~~~~~+\Delta \kappa_i^{\bar{G}}(t) \psi(y(t-\hat{\tau}_i))
+\Delta \kappa_i^{\bar{B}}(t) u(t-\hat{\tau}_i)
\Big],
\\
e(t)=C\varepsilon(t).
\end{array}
\end{equation}

According to model \eqref{eq4_3}, the corresponding  matrices in \eqref{Matrix_Notation} are represented as
\begin{equation*}
\label{Matrix_Notation}
\begin{array}{l}
\bar{\Psi}_{11}=\bar A_0^{\rm T}P+P \bar A_0-Y_0+\sum_{i=0}^{\bar k}S_i,

\\

\Psi_{11}=
\begin{bmatrix}
\bar{\Psi}_{11} & P(\bar A_1-Y_1) &...& P(\bar A_k-Y_{\bar k})
\\
* & -S_1-Y_1 &...& 0
\\
\vdots & \vdots & \ddots & \vdots
\\
* & * & ... & -S_{\bar k}-Y_{\bar k}
\end{bmatrix},

\\

\Psi_{12}=
\begin{bmatrix}
P \bar D_0 & P \bar D_1 &...& P \bar D_{\bar k}
\\
* & 0 &...& 0
\\
\vdots & \vdots & \ddots & \vdots
\\
* & * & ... & 0
\end{bmatrix}.
\end{array}
\end{equation*}
The structure of $\Psi$ is the same as in \eqref{Matrix_Notation}.

\begin{theorem}
\label{Th2}
Let Assumptions \ref{As0}, \ref{As03}--\ref{As6} hold and let there exist matrices $P=P^{\rm T}>0$, $S_i>0$, $i=1,...,\bar{k}$ such that

\begin{equation}
\label{eq_Th011}
\begin{array}{c}
\Psi < 0
~~~\mbox{and}~~~~
PT_0=C^{\rm T}.
\end{array}
\end{equation}
Then the identification algorithms

\begin{equation}
\label{eq_Th022}
\begin{array}{l}
\dot{\hat \kappa}_i^{\bar{A}}(t)^{\rm T}=\Gamma_i^{\bar{A}} \hat{x}(t-\hat{\tau}_i)e(t),
\\
\dot{\hat \kappa}_i^{\bar{D}}(t)^{\rm T}=\Gamma_i^{\bar{D}} \varphi(\hat{x}(t-\hat{\tau}_i))e(t),
\\
\dot{\hat \kappa}_i^{\bar{G}}(t)^{\rm T}=\Gamma_i^{\bar{G}} \psi(y(t-\hat{\tau}_i))e(t),
\\
\dot{\hat \kappa}_i^{\bar{B}}(t)^{\rm T}=\Gamma_i^{\bar{B}} u(t-\hat{\tau}_i)e(t)
\end{array}
\end{equation}
ensure objective \eqref{eq2_3}, where $\Gamma_i^{\bar{A}}$, $\Gamma_i^{\bar{D}}$, and $\Gamma_i^{\bar{G}}$ are positive definite matrices with appropriate dimensions and $\Gamma_i^{\bar{B}}>0$.

\end{theorem}

\begin{proof} It is clear that Theorem \ref{Th1}  is applicable to system  \eqref{eq4_3}, \eqref{eq_Th022} of the same structure as that of \eqref{eq3_2}, \eqref{eq_Th02}. Thus, by applying Theorem \ref{Th1},   the assertion of  Theorem \ref{Th2} is verified. \end{proof}

\begin{remark}
Model \eqref{eq4_0} has a rough approximation relatively to value of $\bar{k}$. Thus, an overestimated number of estimated parameters is in play, and hence,  a larger transient time is obtained. However, using the model 

\begin{equation}
\label{eq_ex001}
\begin{array}{l}
\dot{x}(t)=\sum_{i=0}^{\bar{k}_1}\bar{A}_i x(t-\hat{\tau}_i)
+\sum_{i=0}^{\bar{k}_2}\bar{D}_i \varphi(x(t-\hat{\tau}_i))
\\
~~~~~~~~+\sum_{i=0}^{\bar{k}_3}\bar{G}_i \psi(y(t-\hat{\tau}_i))
+\sum_{i=0}^{\bar{k}_4}\bar{B}_i u(t-\hat{\tau}_i),
\\
y(t)=Cx(t).
\end{array}
\end{equation}
with smaller numbers $\bar{k}_j < \bar{k}$, $j=1,...,4$ of estimated parameters allows one to reduce the number of adjustable parameters, thereby reducing the transient time of estimation of unknown parameters.
It is clear that the algorithm for model \eqref{eq_ex001} remains similar to the algorithm for model \eqref{eq4_0}.
\end{remark}

\section{Example} \label{sec5}

Let model \eqref{eq2_1} be described as
\begin{equation}
\label{eq_ex1}
\begin{array}{l}
\dot{x}(t)=
\begin{bmatrix}
0 & 1
\\
a_{01} & a_{02}
\end{bmatrix}
x(t)+
\begin{bmatrix}
0 & 0
\\
a_{11} & a_{12}
\end{bmatrix}
x(t-\tau_1)
\\
~~~~~~~~+\begin{bmatrix}
0 & 0
\\
d_{11} & d_{12}
\end{bmatrix}
\varphi(x(t-\tau_2))
\\
~~~~~~~~+
\begin{bmatrix}
0
\\
g_{0}
\end{bmatrix}
\psi(y(t))
+
\begin{bmatrix}
0
\\
b_{0}
\end{bmatrix}
u(t)
+
\begin{bmatrix}
0
\\
b_{1}
\end{bmatrix}
u(t-\tau_3),
\\
y(t)=[1~~3]x(t),
\end{array}
\end{equation}
where $x(t)=col\{x_1(t),x_2(t)\}$, the nonlinearities $\varphi(x)=col\{x_1^{\frac{1}{3}}, x_2^{\frac{1}{3}}\}$ and $\psi(y)=y^2$ are known. Only output $y(t)$ and input $u(t)$ are available for measurement. Assume that the value set of the system delays is {\it a priori} known, but it is unknown which delay corresponds to each component $x(t)$, $\varphi(x(t))$, $\psi(y(t))$, $u(t)$. Therefore, according to model \eqref{eq4_0}, rewrite \eqref{eq_ex1} in the form
\begin{equation}
\label{eq_ex2}
\begin{array}{l}
\dot{x}(t)=
\sum_{i=0}^{3}\Big(\begin{bmatrix}
0 & 1
\\
\bar{a}_{i1} & \bar{a}_{i2}
\end{bmatrix}
x(t-\hat{\tau}_i)
\\
~~~~~~~~+
\begin{bmatrix}
0 & 0
\\
\bar{d}_{i1} & \bar{d}_{i2}
\end{bmatrix}
\varphi(x(t-\hat{\tau}_i))
\\
~~~~~~~~+
\begin{bmatrix}
0
\\
\bar{g}_{i}
\end{bmatrix}
\psi(y(t-\hat{\tau}_i))
+
\begin{bmatrix}
0
\\
\bar{b}_{i}
\end{bmatrix}
u(t-\hat{\tau}_i)
\Big),
\end{array}
\end{equation}
where $\hat{\tau}_0=0$, $\hat{\tau}_1=\tau_1$, $\hat{\tau}_2=\tau_2$ and $\hat{\tau}_3=\tau_3$ due to known values of delays.
Thus, model \eqref{eq_ex2} contains any combination of delays in
\eqref{eq_ex1}.

Let $u(t)=\sin(2.3t)+\sin(10t)+\sin(20.2t)+\sin(35.7t)+\sin(51.9t)+P$, $P$ is the function describing pulse generator with amplitude 1, period 1 s and pulse wight $0.5 \%$,  $\tau_1=1$, $\tau_2=1.7$, and $\tau_3=2.3$ in \eqref{eq_ex1},
$\Gamma_i^{\bar{A}}=400I$, $\Gamma_i^{\bar{D}}=400I$, $\Gamma_i^{\bar{G}}=400I$, and $\Gamma_i^{\bar{B}}=400$,
$i=0,...,3$ in \eqref{eq_Th022}.
The simulations show that Assumption \ref{As5} holds for $C \geq 10^3$ and $\alpha \leq 10^{-4}$.
Choosing
$\bar{A}_0^0=\begin{bmatrix} 0 & 1 \\ 0 & 0 \end{bmatrix}$,
$\bar{A}_j^0=\begin{bmatrix} 0 & 0 \\ 0 & 0 \end{bmatrix}$, $j=1,2,3$,
$\bar{D}_i^0=\begin{bmatrix} 0 & 0 \\ 0 & 0 \end{bmatrix}$,
$\bar{G}_i^0=\bar{B}_i^0=\begin{bmatrix} 0 \\ 0 \end{bmatrix}$, $i=0,...,3$,
and $T_0 =[0~~1]^{\rm T}$, Assumption \ref{As6} holds.
Denote
$\kappa_i^{\bar{A}}(t)=[\hat{a}_{i1}(t), \hat{a}_{i2}(t)]$,
$\kappa_i^{\bar{D}}(t)=[\hat{d}_{i1}(t), \hat{d}_{i2}(t)]$,
$\kappa_i^{\bar{G}}(t)=\hat{g}_{i}(t)$, and
$\kappa_i^{\bar{B}}(t)=\hat{b}_{i}(t)$,
where $\hat{a}_{i1}(t)$, $\hat{a}_{i2}(t)$, $\hat{d}_{i1}(t)$, $\hat{d}_{i2}(t)$, $\hat{g}_{i}(t)$, $\hat{b}_{i}(t)$ are the estimates of $\bar{a}_{i1}$, $\bar{a}_{i2}$, $\bar{d}_{i1}$, $\bar{d}_{i2}$, $\bar{g}_{i}$, and $\bar{b}_{i}$ $i=0,...,3$ accordingly.
In Figures the transients of these estimates are presented.

\begin{figure}[h!]
\center{\includegraphics[width=0.65\linewidth]{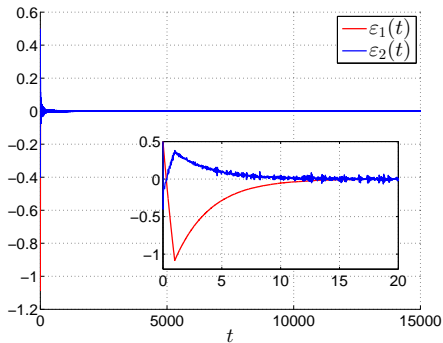}}
\caption{The transients of $\varepsilon(t)=col\{\varepsilon_1(t), \varepsilon_2(t)\}$.}
\label{Fig_A}
\end{figure}

\begin{figure}[h!]
\center{\includegraphics[width=0.65\linewidth]{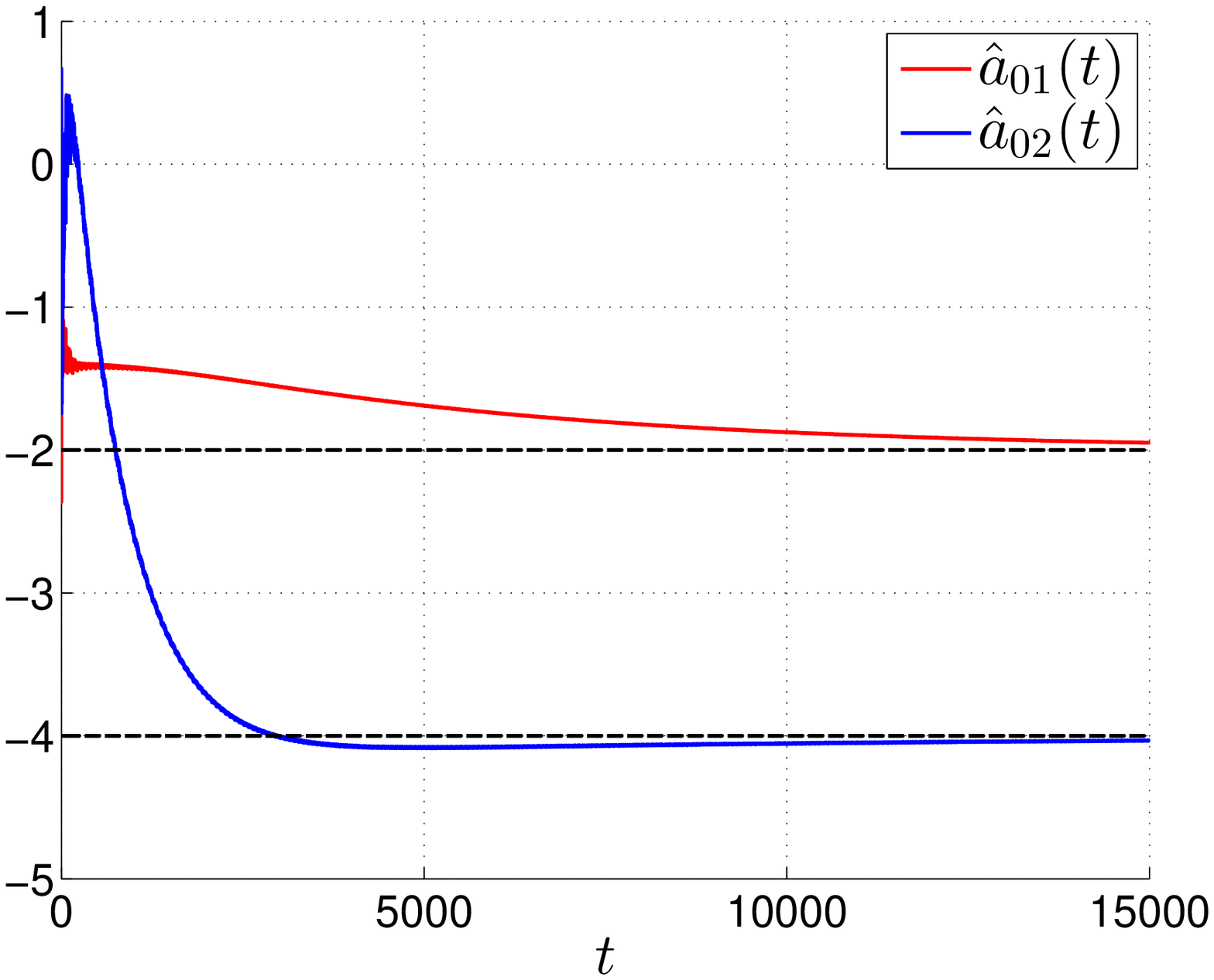}} \\
\center{\includegraphics[width=0.65\linewidth]{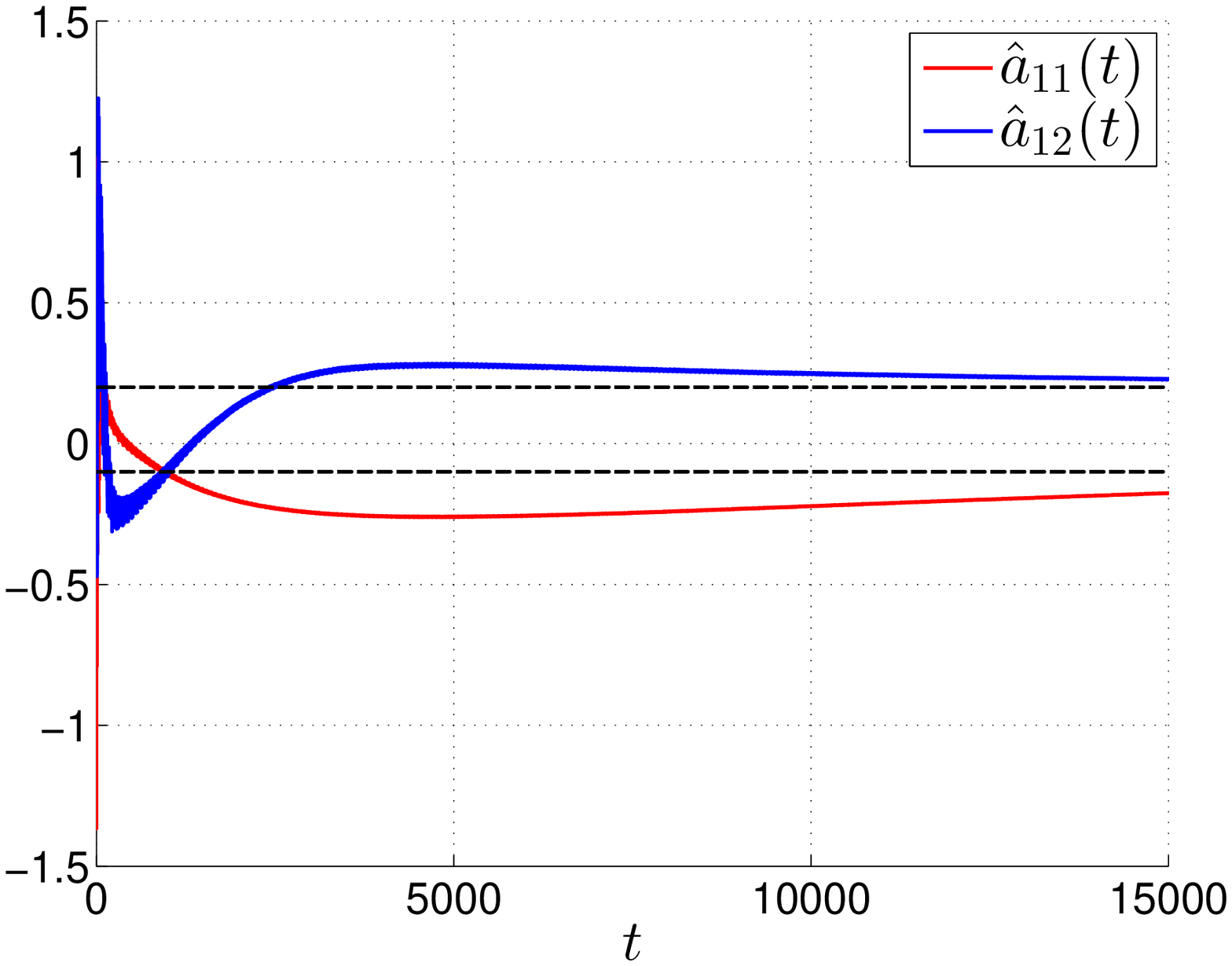}} \\
\center{\includegraphics[width=0.65\linewidth]{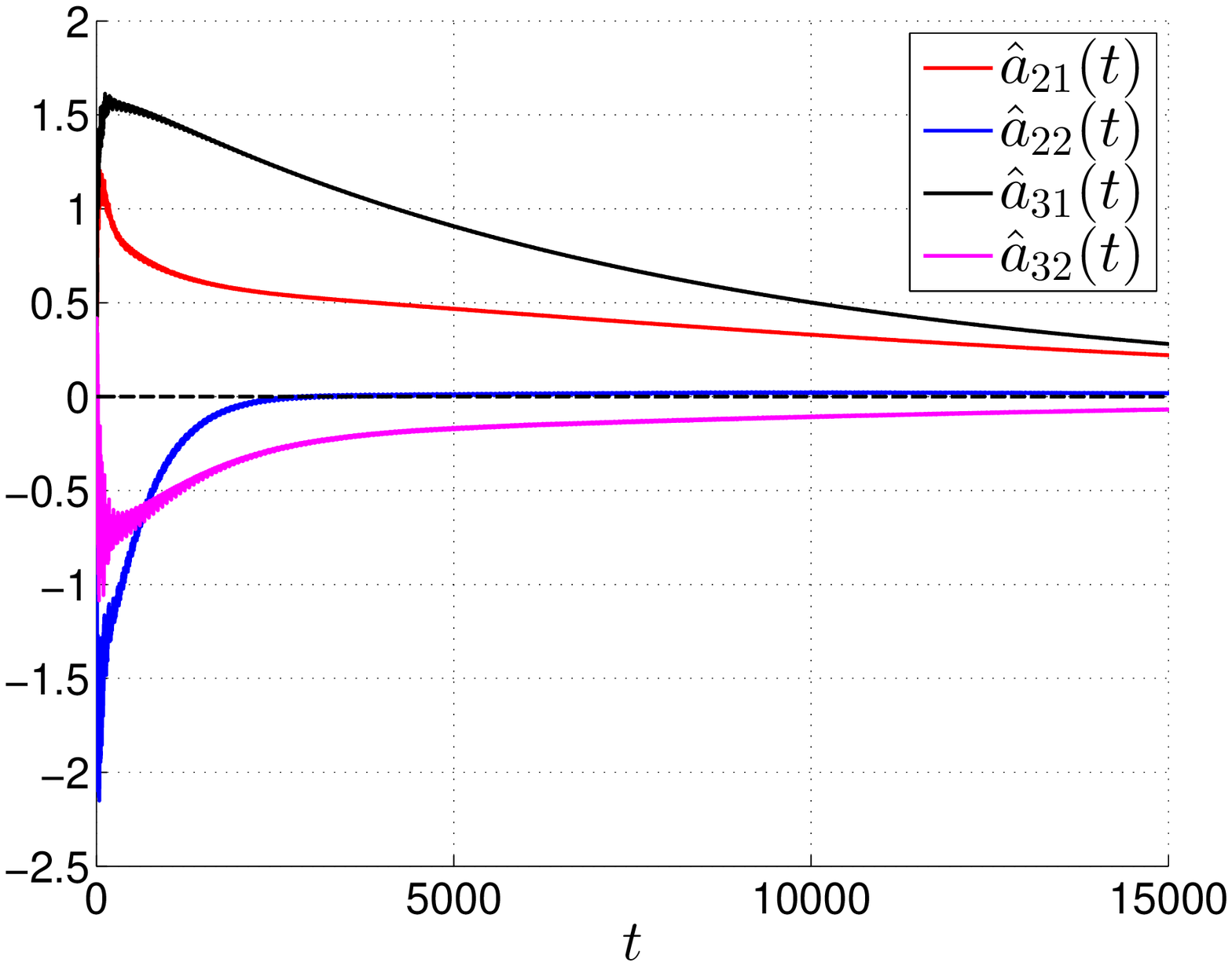}} \\
\caption{The transients of $\hat{a}_{i1}(t)$, $\hat{a}_{i2}(t)$, $i=0,...,3$, where $a_{01}=-2$, $a_{02}=-4$, $a_{11}=-0.1$, $a_{12}=0.2$, $a_{21}=a_{22}=a_{31}=a_{32}=0$.}
\label{Fig_A}
\end{figure}

\begin{figure}[h!]
\center{\includegraphics[width=0.65\linewidth]{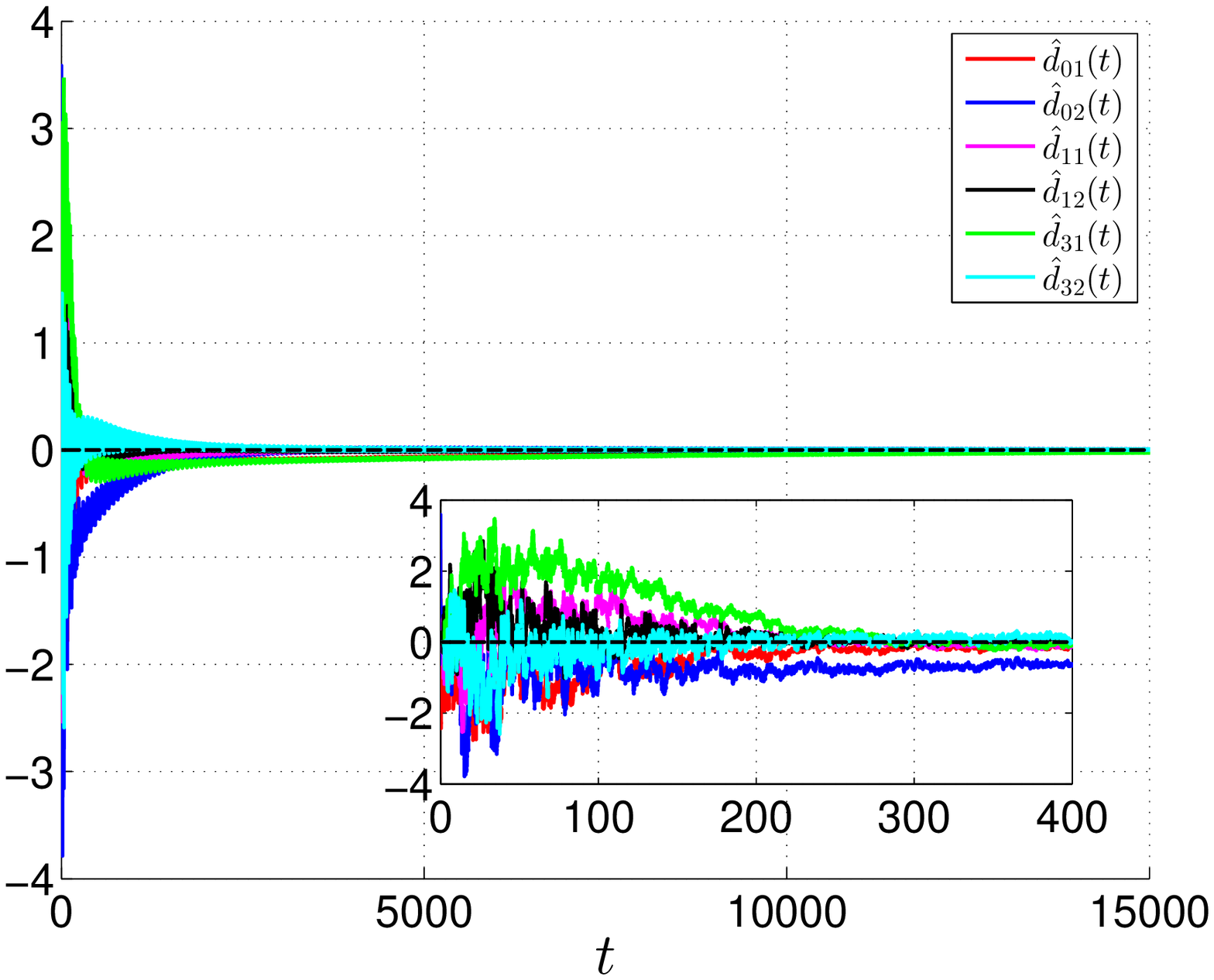}} \\
\center{\includegraphics[width=0.65\linewidth]{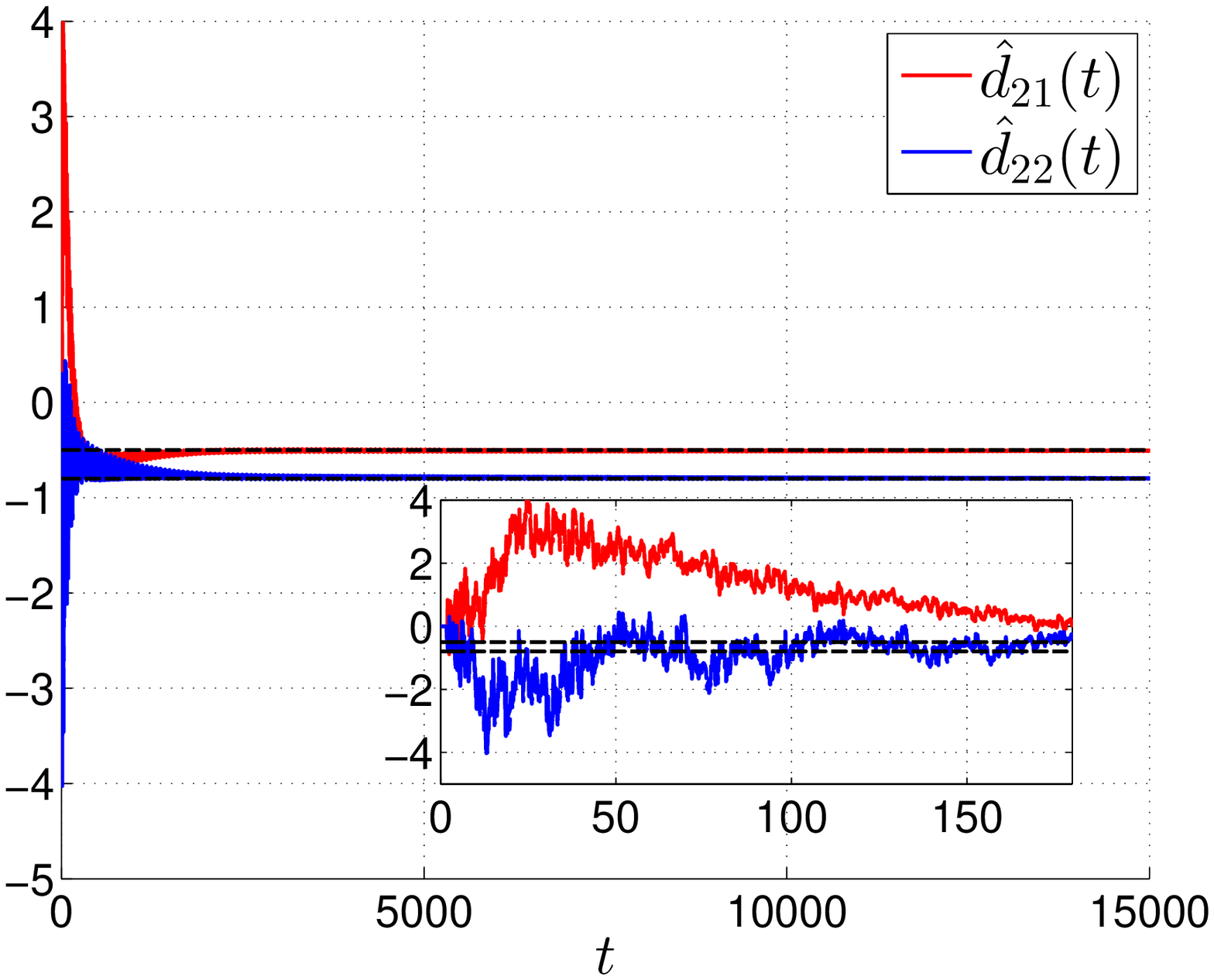}} \\
\caption{The transients of $\hat{d}_{i1}(t)$, $\hat{d}_{i2}(t)$, $i=0,...,3$, where $d_{21}=-0.5$, $d_{22}=-0.8$, $d_{01}=d_{02}=d_{11}=d_{12}=d_{31}=d_{32}=0$.}
\label{Fig_D}
\end{figure}

\begin{figure}[h!]
\center{\includegraphics[width=0.65\linewidth]{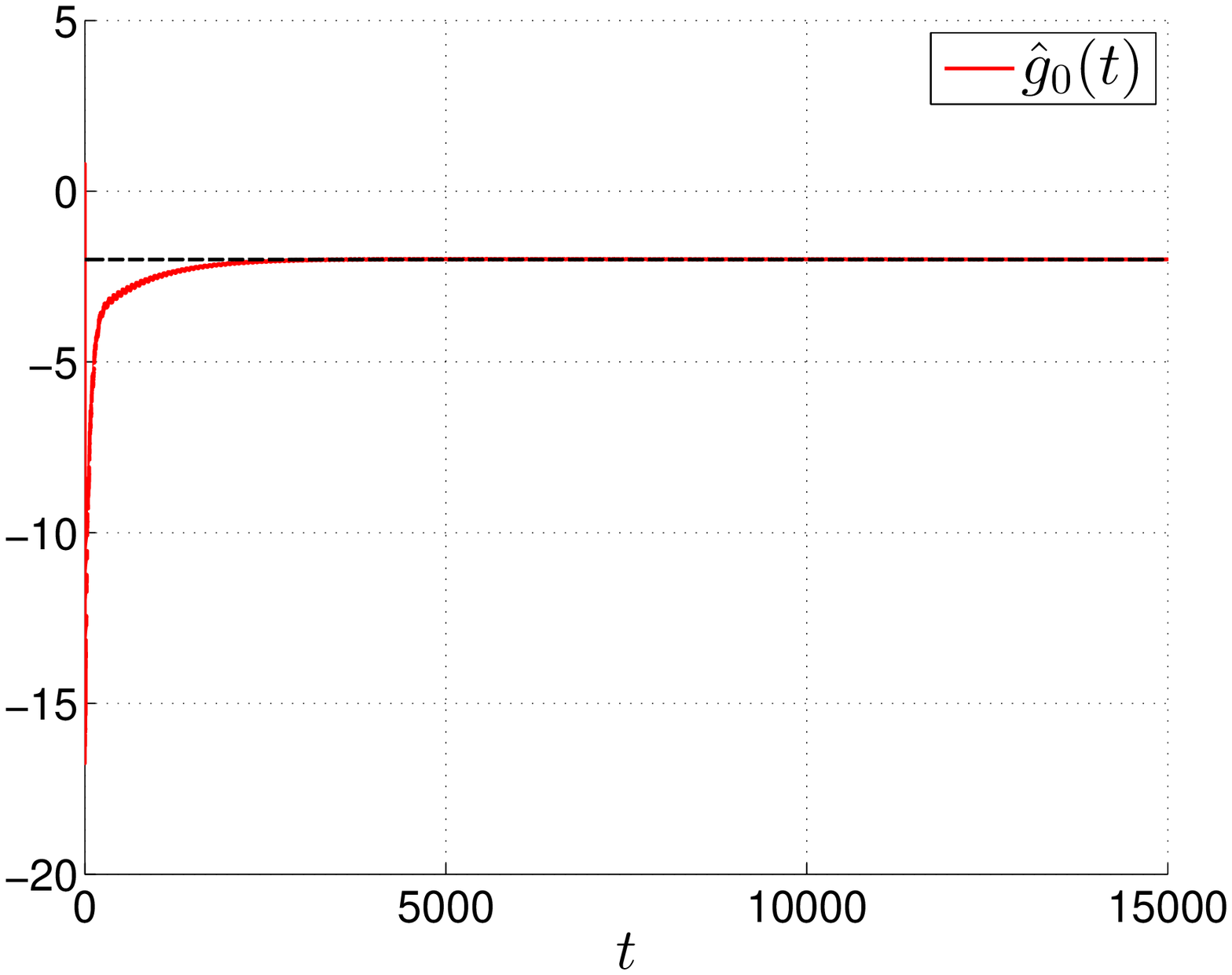}} \\
\center{\includegraphics[width=0.65\linewidth]{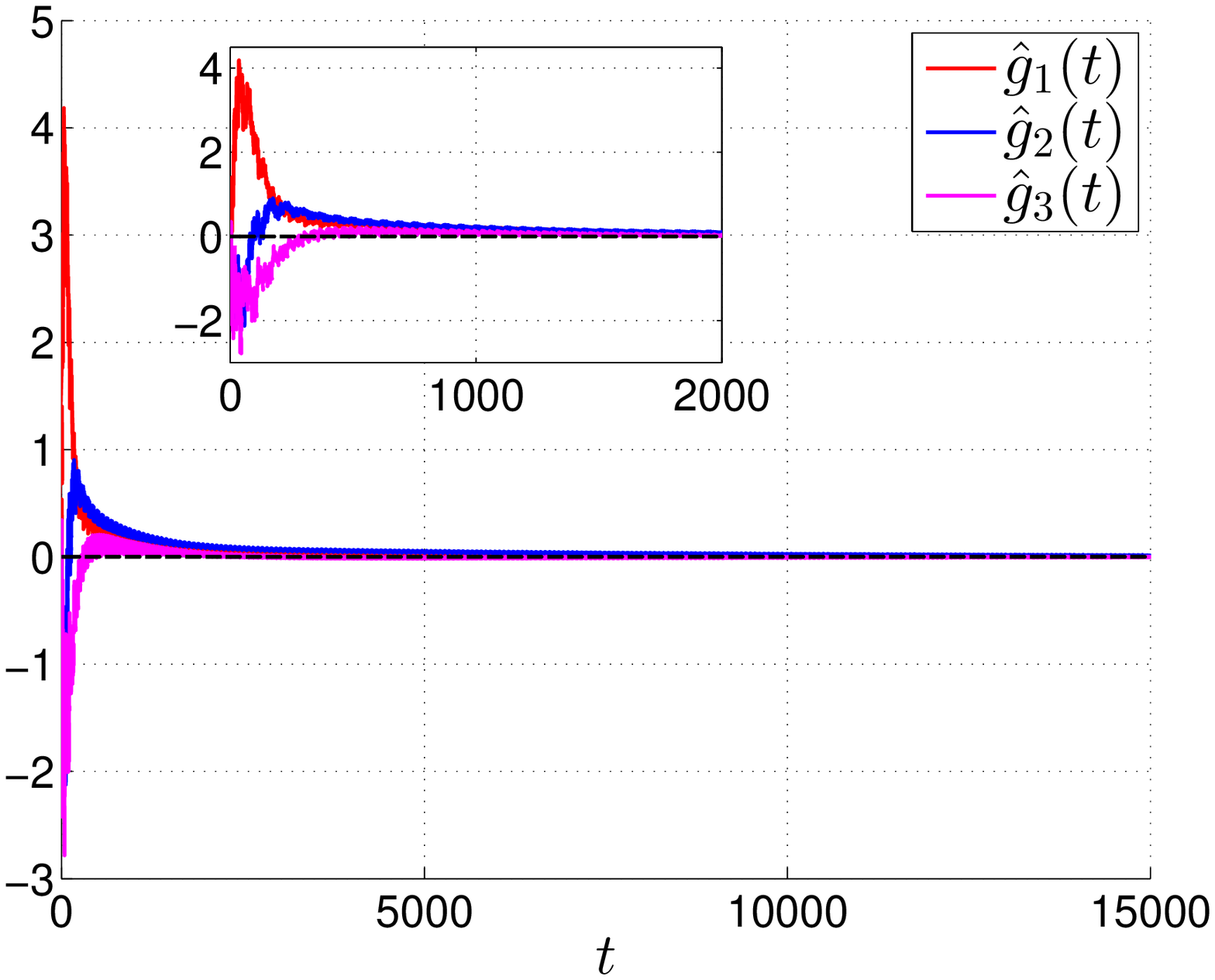}} \\
\caption{The transients of $\hat{g}_{i}(t)$, $i=0,...,3$, where $g_{0}=-2$, $g_{1}=g_{2}=g_{3}=0$.}
\label{Fig_G}
\end{figure}

\begin{figure}[h!]
\center{\includegraphics[width=0.65\linewidth]{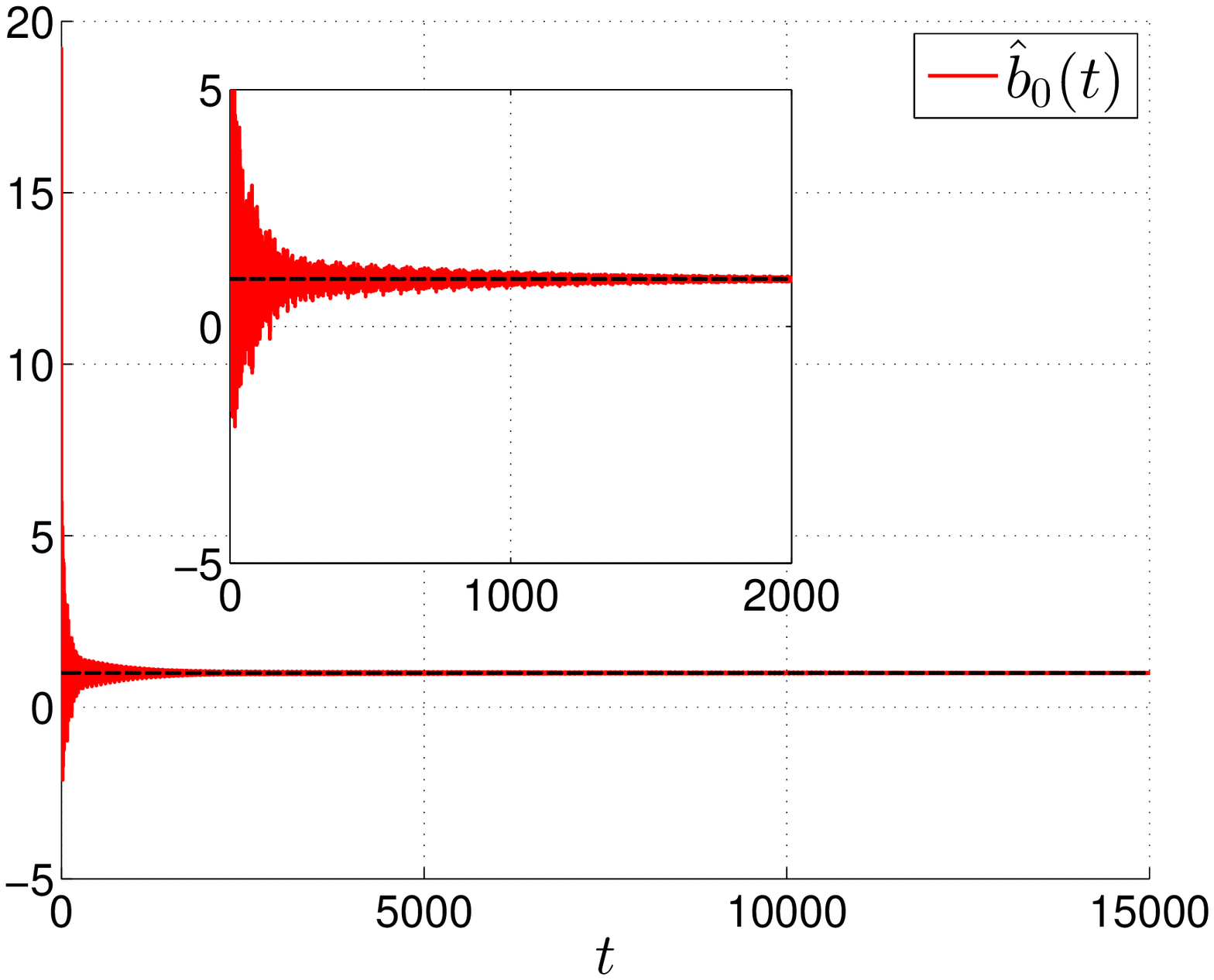}} \\
\center{\includegraphics[width=0.65\linewidth]{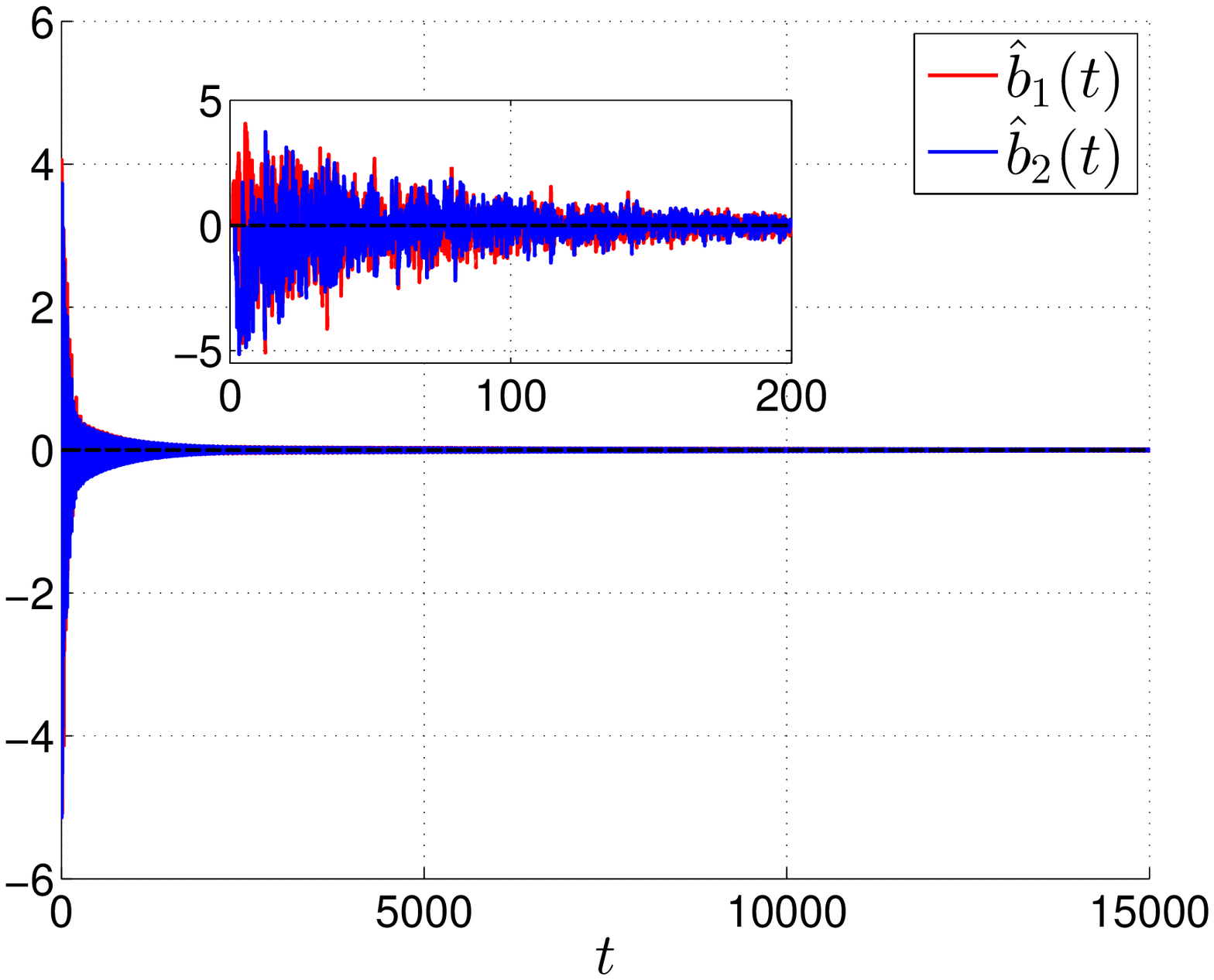}} \\
\center{\includegraphics[width=0.65\linewidth]{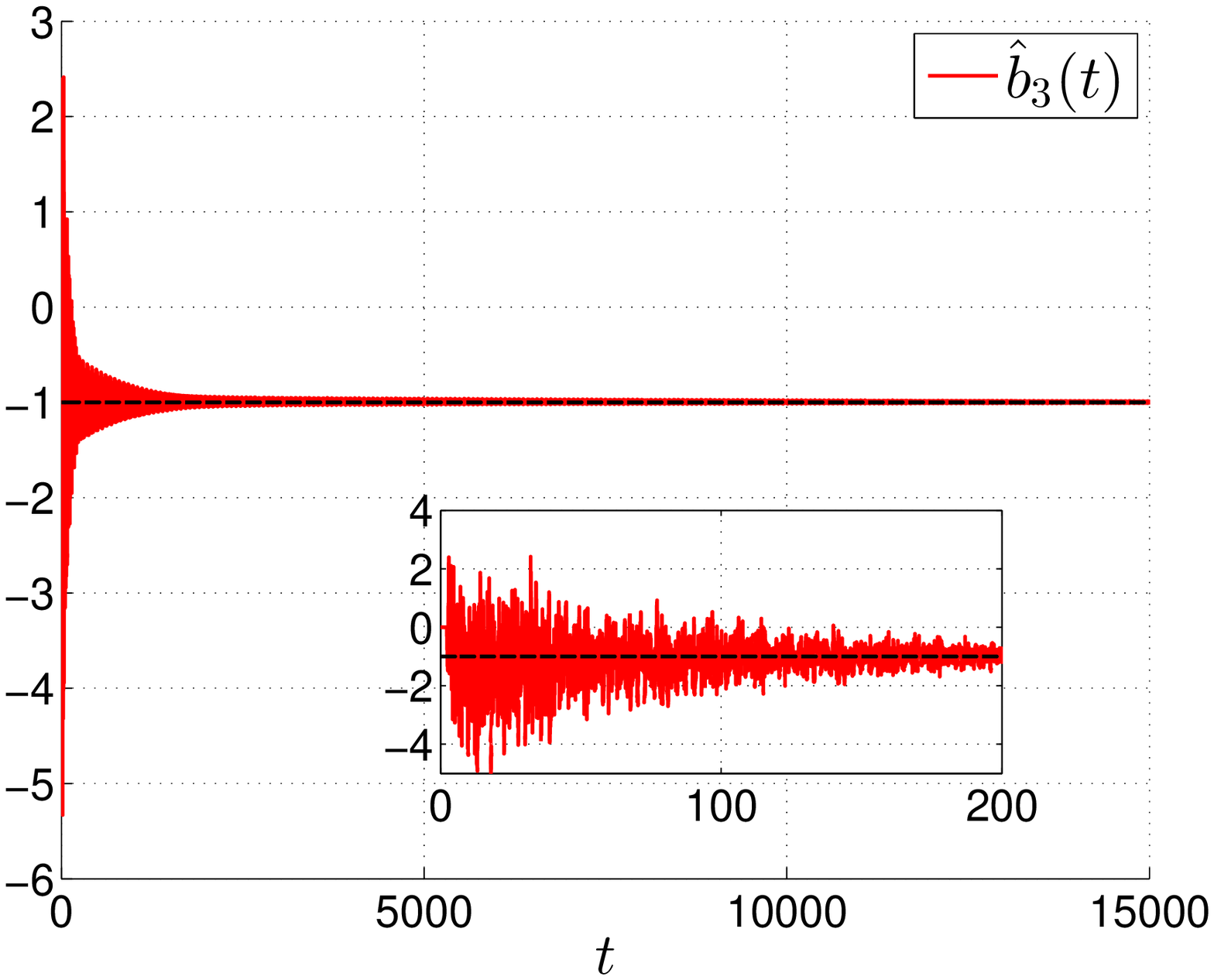}} \\
\caption{The transients of $\hat{b}_{i}(t)$, $i=0,...,3$, where $b_{0}=1$, $b_{1}=b_{2}=0$, and $b_{3}=-1$.}
\label{Fig_B}
\end{figure}

\section{Conclusions}
\label{sec6}

In the paper, a novel adaptive identifier design is proposed for nonlinear systems composed of linear part, Lipschitz and non-Lipschitz nonlinearities. The case of known time-delay values and that of unknown  delays are addressed side by side. In contrast to the existing literature, SISO time delay systems are considered in the general form rather than in the canonical form only. The  identifiability and observability properties are coupled to  the persistent excitation of the plant model to ensure  the asymptotic convergence of estimated parameters to their real values by using the gradient algorithm.
The stability analysis is given in terms of the feasibility  of certain linear matrix inequalities, relying on input and output matrices.
The numerical simulations confirm theoretical results and illustrate efficiency of the proposed algorithm for on-line simultaneous estimation of a large number of unknown parameters, including $2$ state components and $24$ parameters.

\end{document}